\documentclass[letterpaper, 10 pt, conference]{ieeeconf}  
\usepackage{geometry}
\usepackage{cite}
\usepackage{graphicx}
\usepackage{hyperref}
\usepackage{amssymb}
\usepackage{amsmath}   
\usepackage{times}
\usepackage{caption}
\usepackage{flushend}
\usepackage{float}
\usepackage{epstopdf}
\usepackage{subfloat}
\usepackage{subfigure}
\usepackage[section]{placeins}
\usepackage{color}

\usepackage[T1]{fontenc}
\usepackage[utf8]{inputenc}
\usepackage{authblk}

\usepackage{amsthm}
\makeatletter
\def\BState{\State\hskip-\ALG@thistlm}
\makeatother

\newtheorem{theorem}{Theorem}

\newtheorem{lemma}{Lemma}
\newtheorem{assumption}{Assumption}

\newtheorem{remark}{Remark}

\IEEEoverridecommandlockouts                              

\title{\LARGE \bf Finite Time Encryption Schedule in the Presence of an Eavesdropper with Operation Cost }

\author{\thanks{The work by L. Huang and L. Shi is supported by a Hong Kong ITC research fund ITS/066/17FP-A.}Lingying Huang*,\thanks{* Department of Electronic Engineering,
		Hong Kong University of Science and Technology, Clear Water Bay, Kowloon, Hong Kong
		{\tt\small lhuangaq@connect.ust.hk; eesling@ust.hk}	} Alex S. Leong$ ^{+} $\thanks{+ Department of Electrical Engineering (EIM-E), Paderborn University,
		Paderborn, Germany {\tt\small alex.leong@upb.de; dquevedo@ieee.org}}, Daniel E. Quevedo$ ^{+} $, Ling Shi*}

\begin{document}

	\maketitle
	\thispagestyle{empty}
	\pagestyle{empty}
	
	\begin{abstract}
		In this paper, we consider a remote state estimation problem in the presence of an eavesdropper. A smart sensor takes measurement of a discrete linear time-invariant (LTI) process and sends its local state estimate through a wireless network to a remote estimator.  An eavesdropper can overhear the sensor transmissions with a certain probability. To enhance the system privacy level, we propose a novel encryption strategy to minimize a linear combination of the expected error covariance at the remote estimator and the negative of the expected error covariance at the eavesdropper, taking into account the cost of the encryption process.  We prove the existence of an optimal deterministic and Markovian policy for such an encryption strategy over a finite time horizon. Two situations, namely, with or without knowledge of the eavesdropper estimation error covariance are studied and the optimal schedule is shown to satisfy the threshold-like structure in both cases. Numerical examples are given to illustrate the results.
		
	\end{abstract}

\section{Introduction}
Cyber-physical systems (CPSs) integrate sensing, computing and communication capabilities  with physical systems\cite{poovendran2012special}. The introduction of a wireless network enables CPSs to be applied to a wide range of applications. However, it also introduces more challenges to protect privacy. Since information in CPSs is transmitted through unprotected wireless networks in most cases, CPSs are often vulnerable to unauthorized users including malicious attackers. A leakage of confidential information will result in severe consequences, e.g., disclosure of customers' privacy and economic losses\cite{langner2011stuxnet},\cite{cardenas2008secure}.

The most common method to improve system confidentiality is encrypting transmitted packets, e.g., symmetric-key encryption and public-key encryption. Only the legitimate user has the ability to decrypt messages, blocking the access from other adversaries. Reason \cite{reason2001end} proposed that encrypted information should satisfy the avalanche effect property. This property leads to the increase of the average mean squared error at the legitimate receiver as it enlarges the one-bit-error. In addition, cryptography requires more storage and computation services, adding extra burden. Hence, there is a trade-off between the privacy level and the estimation quality as well as a privacy-preserving cost. 

A fairly large body of literature exists in studying the independent and identically distributed  (i.i.d.) packets losses. In this case, the throughput in  wireless network serves as an evaluation indicator in legitimate estimation quality. Haleem et al. \cite{haleem2007opportunistic} presented a mathematical model to capture the security-throughput trade-off. Aysal and Barner \cite{aysal2008sensor} derived an optimal estimator of a deterministic signal using stochastic bit flipping and analyzed the effect. 

On the other hand, it is more general and more difficult to consider that collected packets are measurement vectors of a dynamical system when there is an eavesdropper. Wiese et al. \cite{wiese2016secure} showed that by applying sufficiently large coding length, one could make the estimation error of the eavesdropper unbounded while the legal receiver still has a bounded error covariance for unstable systems (perfect secrecy). Tsiamis et al. \cite{tsiamis2017state} concluded that by exploiting packet erasures policy, perfect secrecy is achieved when the arrival rate of the legitimate receiver is larger than that of eavesdroppers. They also showed in \cite{tsiamis2017state2} that perfect secrecy is achieved with at least one occurrence of the essential event, when the user receives the packet while the eavesdropper fails to intercept it. Furthermore, Leong et al. \cite{leong2017remote} proposed a policy to erase packets based on the estimation error where the system can achieve perfect secrecy even when the eavesdropper has greater probability to obtain information.

Different from \cite{leong2017remote}, we study a more general encryption strategy. We formulate a novel mathematical model to illustrate the effect of encryption strategy (Fig. \ref{fig_sim}) considering a remote estimator and an eavesdropper. Based on this model, we derive structural results on the optimal encryption schedule with (Theorem \ref{t3}), or without (Theorem \ref{t6}) knowledge of the eavesdropper's estimation error covariance. We also introduce the influence of the encryption cost into this optimization problem. With more decision variables, we prove that the threshold structure still holds in both situations (Theorems \ref{t3},\ref{t6}).

This paper is organized as follows. Section \ref{s2} establishes the system model. After analyzing the remote estimator's and the eavesdropper's performance, we introduce the mathematical formulation of the main problem. Section \ref{se3} proves the existence and the structure of an optimal deterministic and Markovian policy in a finite time horizon with or without knowledge of the eavesdropper's estimation error covariance. Numerical simulations are given in Section \ref{s4}. Section \ref{s5} draws conclusions.


\textit{Notation:} $ \mathbb{N} $ is the set of natural numbers. $ \mathbb{R} $ and $ \mathbb{R}^{n} $ represent the set of real numbers and $ n- $dimensional real column vectors. For a matrix $ X $, $  X' $ and $ tr(X)$
denote its transpose and trace, respectively. When $ X $ is a positive semidefinite
matrix, it is written as $ X\geq0 $. The notation $ \mathbb{P}(\cdot) $ and $  \mathbb{E}[\cdot]  $ are the probability and expectation  of a random matrix, respectively,
and  $ \mathbb{E}[\cdot|\cdot]  $ is its conditional expectation. 
 For functions $ f ,f_{1}$ and $f_{2} $, $ f_{1}\circ f_{2} $ is defined
as $ f_{1}\circ f_{2}=f_{1}( f_{2}(X)) $  and $ f^{ k} $ is defined as $ f^{ k}(X)=\underbrace{f\circ f\circ \ldots \circ f}_{k\text{ times}}(X)  $,
 with $ f^{ 0}(X) = X $. A function $ F(\cdot) $ is increasing if $ X\leq Y \Rightarrow F(X)\leq F(Y)$. A function $ F(\cdot) $ is decreasing if $ X\leq Y \Rightarrow F(X)\geq F(Y)$.

\section{System model and problem formulation}\label{s2}
\begin{figure}[!t]
\centering
\includegraphics[width=3.3in]{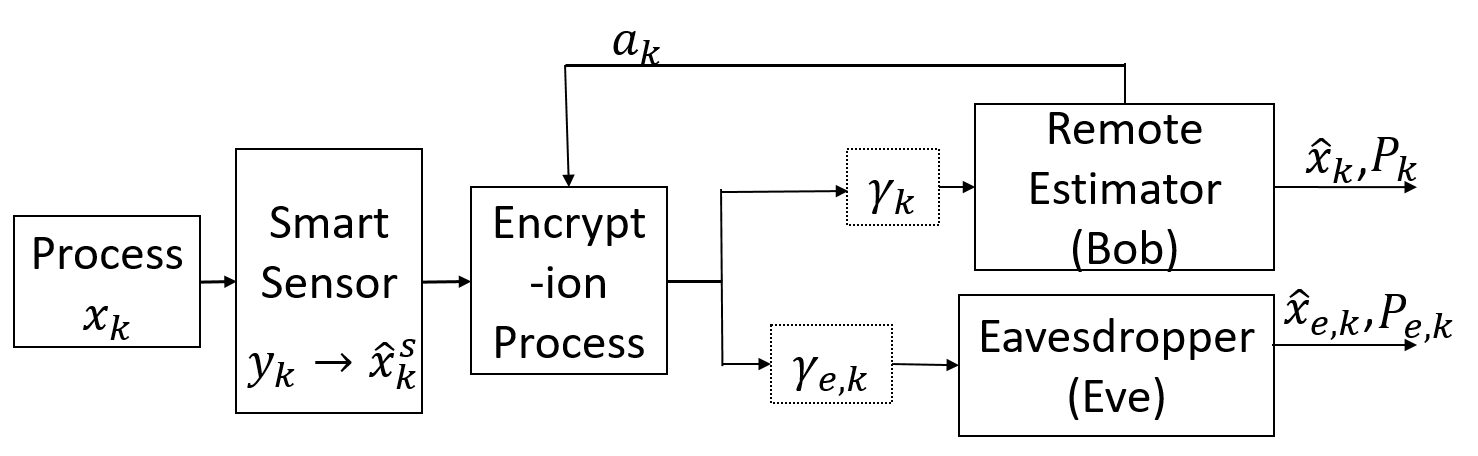}
\caption{System structure.}
\label{fig_sim}
\end{figure}

\subsection{System Model}
Consider the linear time-invariant (LTI) system in Fig 1, which is given as follows
\begin{equation}
\begin{split}
x_{k+1}&=Ax_{k}+w_{k},\\
y_{k}&=Cx_{k}+v_{k},
\end{split}
\end{equation}
	where $ k\in \mathbb{N} $ is the time index, $ x_{k}\in \mathbb{R}^{n} $ is the system state, $  y_{k} \in{\mathbb{R}^{m}}  $ is the measurement vector taken by the sensor at
time $ k $, $ w_{k} \in \mathbb{R}^{n} $ and $ v_{k} \in \mathbb{R}^{m} $ are two i.i.d. zero-mean Gaussian random
noises with covariances $ Q \geq 0$ and $ R > 0 $, respectively.
The initial state $  x_{0} $ is a zero-mean Gaussian random vector that
is uncorrelated with $ w_{k} $ or $ v_{k} $, and has covariance $ \Pi_{0} \geq 0 $. We further
assume that ($ A $,$ \sqrt{Q} $) is controllable and ($  A, C $) is observable.

A smart sensor equipped with computation and memory capacity is capable of running a local Kalman filter. The sensor transmits quantities $ \hat{x}^{s}_{k} $ to a remote estimator (Bob). 
According to Anderson and Moore \cite{anderson1979optimal}, since ($ A $,$ \sqrt{Q} $) is controllable and ($ A, C $) is observable, posteriori estimation error covariance $ P_{k}^{s}  $ converges exponentially fast to a steady state $ P^{*} $. For simplicity,  we assume  
$ 
P_{k}^{s}=P^{*}$.

Let $ a_{k} \in \{0, 1\}  $ be a decision variable. When $ a_{k}=0 $, the sensor transmits its local state estimate $ \hat{x}_k^{s} $  to the remote estimator. Otherwise, when $ a_{k}=1 $,  the local estimate $ \hat{x}_k^{s} $  is first encrypted before transmitting. 
The decision variable $ a_{k} $ is determined at the remote estimator, which
is assumed to have more computational capabilities than the
sensor. It uses the information available at time $ k-1 $, and then feeds
back to the sensor before transmission at time $ k $.

 We use $ \gamma_{k} $ to represent whether the remote estimator receives $ \hat{x}_k^{s} $ successfully at time $ k $, i.e., $ \gamma_{k}=1 $ indicates that the local estimate is received successfully by the remote estimator at time $ k $ and $ \gamma_{k}=0 $ otherwise. We make the following assumption about effects of the encryption process on the packet arrival rate.

\begin{assumption}
	The packet arrival rate is memoryless and is only affected if the transmitted messages are encrypted,  i.e.,  $ \{\gamma_{k}\}$ is i.i.d.. It is assumed that encryption contributes an impact factor $ \epsilon_{1}  (0\leq \epsilon_{1}\leq 1)  $ to the arrival rate. The following equality holds for any $ k\geq 1 $,
	\begin{equation}
	\begin{split}
	\mathbb{P}(\gamma_{k}=1)=\left\lbrace
	\begin{array}{l}
	\lambda, \text{ if }a_{k}=0 ,\\
	\epsilon_{1}\lambda, \text{ if }a_{k}=1.
	\end{array}\right.
	\end{split}
	\end{equation}
	\begin{remark}
		The impact factor can be determined by the type of the encryption. A large number of published studies focus on specific impact factor of different encryption methods. For example, in paper \cite{leong2017remote} , the impact factor of packet erasure method is 0. Meanwhile, paper \cite{reason2001end} showed that the perceptual degradation in subjective quality caused by confidentiality closely follows the quantitative degradation in bit-error rate. Therefore, if the packet length is known to the remote estimator in advance, the impact factor is deterministic and we simplify it as a general constant $ \epsilon_{1} $.
	\end{remark}
\end{assumption}

There exists an eavesdropper (Eve) who can overhear the sensor transmission. Let $ \gamma_{e,k} $ be a random variable such that $ \gamma_{e,k}=1 $ if $ \hat{x}_k^{s} $ is overheard and decrypted successfully by the eavesdropper, and $ \gamma_{e,k}=0 $ otherwise.  We make the following assumption about the influence of successful eavesdropping rate (which means the message is obtained and decrypted successfully) at the eavesdropper.
\begin{assumption}
	The successful eavesdropping rate for the eavesdropper is memoryless. If the message is encrypted, the eavesdropper has fixed probability $ \epsilon_{2}  (0\leq \epsilon_{2}\leq 1) $ to decrypt it. Therefore, the following equality holds for $ k\geq 1 $,
	\begin{equation}
	\begin{split}
	\mathbb{P}(\gamma_{e,k}=1)=\left\lbrace
	\begin{array}{l}
	\lambda_{e}, \text{ if }a_{k}=0, \\
	\epsilon_{2}\lambda_{e}, \text{ if }a_{k}=1.
	\end{array}\right.
	\end{split}
	\end{equation}
	The processes $ \{\gamma_{k}\}$ and  $ \{\gamma_{e,k}\}$  are assumed to be mutually independent. 	
	\begin{remark}
		To make the decryption probability memoryless, if we use key to encrypt the messages, we need to change the key from time to time from being deciphered by the eavesdropper.
	\end{remark}
\end{assumption}

It is assumed that the remote estimator knows the decision variable $ a_{k} $ and whether the transmission was successful or not, i.e., $ \gamma_{k} $. According to \cite{sinopoli2004kalman}, the remote estimator's state estimate $ \hat{x}_{k} $ and the corresponding error covariance $ P_{k} $ at time $ k $ are given by
\begin{equation}
	\begin{split}
		(\hat{x}_{k},P_{k})=\left\lbrace
		\begin{array}{l}
		(A\hat{x}_{k-1},h(P_{k-1})), \text{ if }\gamma_{k}=0 ,\\
		(\hat{x}^{s}_{k},P^{*}), \text{ if }\gamma_{k}=1,
		\end{array}\right.		
	\end{split}
\end{equation}
where the Lyapunov operator $ h(X)\triangleq AXA+Q $.

Similarly, the eavesdropper knows if it has eavesdropped successfully, i.e., $ \gamma_{e,k} $. The state estimate $ \hat{x}_{e,k} $ and error covariance $ P_{e,k} $ at time $ k $ are
\begin{equation}
\begin{split}
(\hat{x}_{e,k},P_{e,k})=\left\lbrace
\begin{array}{l}
(A\hat{x}_{e,k-1},h(P_{e,k-1})), \text{ if }\gamma_{k}=0 ,\\
(\hat{x}^{s}_{k},P^{*}), \text{ if }\gamma_{k}=1.
\end{array}\right.		
\end{split}
\end{equation}

\begin{lemma} (\cite{shi2010kalman} )
	For any $ k_{1}\geq k_{2}\geq0 $,	$ h^{k_{1}}(P^{*})\geq h^{k_{2}}(P^{*}) $.
	Therefore, $ tr(h^{k_{1}}(P^{*}))\geq tr(h^{k_{2}}(P^{*}) )$.
\end{lemma}

Define  
$ \mathbf{S}\stackrel{\mathrm{\Delta}}{=} \{P^{*},h(P^{*}),h^2(P^{*})\ldots\} $ which  consists of all possible values of $ P_{k} $ and $ P_{e,k} $.
From Lemma 1, there is a total ordering 
$ P^{*}\leq h(P^{*})\leq \cdots $, thus $ \mathbf{S}  $ is a total order set.

\subsection{Problem of interest}
Considering the finite time horizon, our goal is to minimize a linear combination of the expected error covariance at the remote estimator and the
negative of the expected error covariance at the eavesdropper, while taking into account the operation cost of the encryption process. The integrated cost $ J_{k} $ considering the privacy level, the system performance and the encryption cost is 
\begin{equation}\label{e9}
	\min\limits_{a_{k}\in \{0,1\}}J_{k}\triangleq\sum\limits_{k=1}^{N}\mathbb{E}[\beta tr(P_{k})-(1-\beta)tr(P_{e,k})+a_{k}\mathcal{C}].
\end{equation} 
The coefficient $ \beta\in(0,1) $ weighs the importance of the error covariance of the system compared with that of the eavesdropper. With larger $ \beta $, it means that maintaining the system performance is of more importance than minimizing the information leakage, and vice versa. The parameter $ \mathcal{C} $ is the  normalized total cost of the encryption process.

\begin{remark}
	Packet erasure presented in paper \cite{leong2017remote} can be viewed as a special encryption strategy in our model with $ \epsilon_{1}=\epsilon_{2}=0 $ and $ \mathcal{C} =0 $. In our subsequent analysis, we will show that optimal policies are still of threshold-type. 
\end{remark}

\section{Finite Time Horizon MDP Framework}\label{se3}
\subsection{Eavesdropper State Known at Remote Estimator}\label{s3}
We first consider the easier case where the eavesdropper error covariance is known at the remote estimator. We derive a discrete time Markov decision process (MDP) problem.
\begin{enumerate}
	\item The state  $s_{k}\triangleq(P_{k-1},P_{e,k-1}) $ at time $ k $  belongs to the state space $ \mathbb{S} \subset \mathbf{S}\times\mathbf{S}$. 
	\item The action $ a_{k}\in\{0,1\} $ belongs to the action space $ \mathbb{A}$.
	\item The state transition probability distribution $\mathbb{P}(s'|s,a) $ is time homogeneous, where  $ s',s\in \mathbb{S}$  , $ a\in\mathbb{A} $ by dropping the time index and $ s' $ is next state when taking action $ a $ at current state $  s  $. Denote $ s_{00}\triangleq(h(P),h(P_{e}) ),s_{01}\triangleq(h(P),P^{*} ),s_{10}\triangleq(P^{*},h(P_{e})),s_{11}\triangleq(P^{*},P^{*})$ and $ s=(P,P_{e}) $, then we obtain
	\begin{equation}
	\begin{split}
	\mathbb{P}_{00}(0)\triangleq\mathbb{P}(s_{00}|s,0)&=(1-\lambda)(1-\lambda_{e}),\\
	\mathbb{P}_{01}(0)\triangleq\mathbb{P}(s_{01}|s,0)&=(1-\lambda)\lambda_{e},\\
	\mathbb{P}_{10}(0)\triangleq\mathbb{P}(s_{10}|s,0)&=\lambda(1-\lambda_{e}),\\
	\mathbb{P}_{11}(0)\triangleq\mathbb{P}(s_{11}|s,0)&=\lambda\lambda_{e},\\
	\mathbb{P}_{00}(1)\triangleq\mathbb{P}(s_{00}|s,1)&=(1-\epsilon_{1}\lambda)(1-\epsilon_{2}\lambda_{e}),\\
	\mathbb{P}_{01}(1)\triangleq\mathbb{P}(s_{01}|s,1)&=(1-\epsilon_{1}\lambda)\epsilon_{2}\lambda_{e},\\
	\mathbb{P}_{10}(1)\triangleq\mathbb{P}(s_{10}|s,1)&=\epsilon_{1}\lambda(1-\epsilon_{2}\lambda_{e}),\\
	\mathbb{P}_{11}(1)\triangleq\mathbb{P}(s_{11}|s,1)&=\epsilon_{1}\lambda \epsilon_{2}\lambda_{e}.
	\end{split}	
	\end{equation}
	\item The one-stage cost function at time $ k $ is
	\begin{equation}
	\begin{split}
		&c_{k}(P_{k-1},P_{e,k-1},a_{k})\triangleq a_{k}\mathcal{C}\\		
		&+\mathbb{E}[\beta tr(P_{k})-(1-\beta)tr(P_{e,k})|P_{k-1},P_{e,k-1},a_{k}]
		\\
		&=a_{k}\mathcal{C}+\beta[(a_{k}\epsilon_{1}\lambda+(1-a_{k})\lambda)tr(P^{*})\\
		&+(1-a_{k}\epsilon_{1}\lambda-(1-a_{k})\lambda)tr(h(P_{k-1}))]\\
		&-(1-\beta)[(a_{k}\epsilon_{2}\lambda_{e}+(1-a_{k})\lambda_{e})tr(P^{*})\\
		&+(1-a_{k}\epsilon_{2}\lambda_{e}-(1-a_{k})\lambda_{e})tr(h(P_{e,k-1}))].
	\end{split}	
	\end{equation}
	\begin{remark}
		From Lemma 1, the one-stage cost function $ c_{k}$ increases in $ P_{k-1} $ and decreases in $ P_{e,k-1} $.
	\end{remark} 
\end{enumerate}

By above definitions, problem \eqref{e9} is equal to 
\begin{equation}\label{e13}
\min\limits_{a_{k}\in \{0,1\}}\sum\limits_{k=1}^{N}	c_{k}(P_{k-1},P_{e,k-1},a_{k}).
\end{equation} 

Define the optimality equation (Bellman equation) as
\begin{equation}\label{e1}
\begin{split}
	&V_{k}(P,P_{e})= \min\limits_{a\in \{0,1\}}\{c_{k}(P,P_{e},a)+\mathbb{P}_{00}(a)V_{k+1}(s_{00})+\\
	&\mathbb{P}_{01}(a)V_{k+1}(s_{01})+\mathbb{P}_{10}(a)V_{k+1}(s_{10})+\mathbb{P}_{11}(a)V_{k+1}(s_{11})\},
\end{split}	
\end{equation}
 where $ V_{k}(\cdotp,\cdotp) $ for $ k=1,2,\ldots,N $ is a real valued function and $ V_{N+1}(P,P_{e})=0 $.

\begin{theorem}\label{t1}
 There exists an optimal deterministic Markovian policy to problem \eqref{e13}.
\end{theorem}
\begin{proof}
	In a finite time horizon, the state set $ \mathbb{S} $ is finite and the corresponding	action set $ \mathbb{A} $ is finite. As the action set $ \mathbb{A} $ is finite,  there always exists a deterministic and Markovian optimal policy \cite{puterman2014markov}. 	
\end{proof}

Let $ H_{k}=(s_{0},a_{1},\ldots,s_{k-1},a_{k},s_{k}) $ stand for the history information up to time $ k $.
We will make the action $ a_{k+1} $ based on $ H_{k} $ to minimize the total integrated cost. The Markovian property determined by Theorem \ref{t1}  guarantees that the future is independent of the past given the present\cite{puterman2014markov}\cite{bertsekas2005dynamic}. Hence, choosing actions $ a_{k+1} $ based on $ s_{k} $ would be the same  as choosing actions based on $ H_{k} $ and problem \eqref{e13} can be solved in a recursive way as
\begin{equation*}
\begin{split}
&V_{N+1}(P,P_{e})=0,\\
&V_{k}(P,P_{e})=\min\limits_{a\in \{0,1\}}\{c_{k}(P,P_{e},a)+\mathbb{P}_{00}(a)V_{k+1}(s_{00})+\\
&\mathbb{P}_{01}(a)V_{k+1}(s_{01})+\mathbb{P}_{10}(a)V_{k+1}(s_{10})+\mathbb{P}_{11}(a)V_{k+1}(s_{11})\}.
\end{split}	
\end{equation*}

 The following theorem will be used to establish that the optimal solution has a threshold property (Theorem \ref{t3}).

\begin{theorem}\label{t2}
	The optimal value function $ V_{k}(P,P_{e}) $ is an increasing function in $ P $ and a decreasing function in $ P_{e} $.  
\end{theorem}
\begin{proof}
 See Appendix \ref{a1}.
\end{proof}

\begin{theorem}\label{t3}
	 	(1) For a fixed $ P_{e,k-1} $, the optimal solution to problem \eqref{e13} is a threshold policy on $ P_{k-1} $
	\begin{equation}
		a_{k}^{*}(P_{k-1},P_{e,k-1})=\left\lbrace
		\begin{array}{l}
			1, \text{if}\, P_{k-1}\leq h^{m(k)}(P^{*}), \\
			0, \text{otherwise }.
		\end{array}\right. 
	\end{equation}
	where the threshold $ m(k)\in \mathbb{N} $ depends on $ k $ and $ P_{e,k-1} $.
	
	(2) For a fixed $ P_{k-1} $, the optimal solution to problem  \eqref{e13}  is a threshold policy on $ P_{e,k-1} $ 
	\begin{equation}
	a_{k}^{*}(P_{k-1},P_{e,k-1})=\left\lbrace
	\begin{array}{l}
	1, \text{if}\, P_{e,k-1}\geq h^{m_{e}(k)}(P^{*}) ,\\
	0, \text{otherwise }.
	\end{array}\right. 
	\end{equation}
	where the threshold $ m_{e}(k)\in \mathbb{N}  $ depends on $ k $ and $ P_{k-1} $.
\end{theorem}
\begin{proof}
See Appendix \ref{a2}.
\end{proof}
\begin{remark}
	Theorem \ref{t3} can be viewed in an intuitive way, i.e., (1) shows that the optimal policy is to transmit the packet without encryption to Bob when $ P_{k-1} $ is large, as we want to reduce $P_{k}$ but the encryption makes the arrival rate smaller. For (2), it can be understand as that it is more efficient to encrypt the message when $ P_{e,k-1} $ is large, since we want $ P_{e,k} $ to increase even further.
\end{remark}

\subsection{Eavesdropper State Unknown at Remote Estimator}\label{un}
In real situations, the malicious eavesdropper would hide itself from being detected by the remote estimator as far as possible. Therefore, it is difficult to know the eavesdropper's error covariance. Here we assume that the remote estimator knows the leakage probability $ \lambda_{e} $ from previous measurements, but is not aware of the actual realization of $ \gamma_{e,k} $. This can be viewed as a partially observable MDP (POMDP) problem. This POMDP can be converted to a completely observable MDP using belief vector states.

Define the belief vector $ \pi_{e,k} $, which represents the probability distribution of $ P_{e,k} $ given the encryption schedule as
\begin{equation}
	\pi_{e,k}\triangleq\begin{bmatrix}
	\pi_{e,k}^{0}\\ 
	\pi_{e,k}^{1}\\ 
	\colon\\ 
	\pi_{e,k}^{N}
	\end{bmatrix}=\begin{bmatrix}
	\mathbb{P}(P_{e,k}=P^{*}|a_{1},\ldots,a_{k})\\ 
	\mathbb{P}(P_{e,k}=h(P^{*})|a_{1},\ldots,a_{k})\\ 
	\colon\\ 
	\mathbb{P}(P_{e,k}=h^{N}(P^{*})|a_{1},\ldots,a_{k})
	\end{bmatrix} .
\end{equation}
Denote the set of all possible $ \pi_{e,k} $'s as $ \Pi_{e}\subseteq R^{N+1} $. By our assumption, we have $ P_{e,0}=P^{*} $ and $ 	\pi_{e,0}=\begin{bmatrix}
1& 0&\ldots& 0\end{bmatrix}^{T} $.

We can obtain a recursive relationship for $\pi_{e,k} $ as
\begin{equation}
\pi_{e,k}=\Phi(\pi_{e,k-1},a_{k}),
\end{equation}
where
\begin{align*}
	&\Phi(\pi_{e},a)\triangleq\\
	&\left\lbrace
	\begin{array}{l}
	\begin{bmatrix}
	\lambda_{e}& (1-\lambda_{e})\pi_{e}^{0}&\ldots& (1-\lambda_{e})\pi_{e}^{N-1}\end{bmatrix}^{T}, \text{ if }a=0 ,\\
	\begin{bmatrix}
	\epsilon_{2}\lambda_{e}& (1-\epsilon_{2}\lambda_{e})\pi_{e}^{0}&\ldots& (1-\epsilon_{2}\lambda_{e})\pi_{e}^{N-1}\end{bmatrix}^{T} ,\text{ if } a=1.
	\end{array}\right. 
\end{align*}

Different from Section \ref{s3}, Bob will make the decision $ a_{k } $ based on its own $ P_{k-1} $ and the belief vector $ \pi_{e,k-1} $ since $ P_{e,k} $ is unknown to Bob.
Therefore, in this subsection, the discrete time MDP problem is the following
\begin{enumerate}
	\item The state $ s_{k}\triangleq(P_{k-1},\pi_{e,k-1}) $  at time $ k $ belongs to the state space $ \mathbb{S}\subset\mathbf{S}\times \Pi_{e} $.
	\item The action $ a_{k}\in\{0,1\} $ is in the action space $ \mathbb{A} $.
	\item Denote $ s\triangleq(P,\pi_{e}) $, $ s'\triangleq(P^{+},\pi_{e}^{+}) $. The state transition probability distribution $\mathbb{P}(s'|s,a) $ is
	\begin{equation}
	\begin{split}
	&\mathbb{P}(P^{+},\pi_{e}^{+}|s,a)=\\
	&\left\lbrace
	\begin{array}{l}
	\lambda,\text{if}\, a=0,P^{+}=P^{*},\text{ if }\pi_{e}^{+}=\Phi(\pi_{e},0),\\
	1-\lambda,\text{if}\, a=0,P^{+}=h(P),\text{ if }\pi_{e}^{+}=\Phi(\pi_{e},0),\\
	\epsilon_{1}\lambda,\text{if}\, a=1,P^{+}=P^{*},\text{ if }\pi_{e}^{+}=\Phi(\pi_{e},1),\\
	1-\epsilon_{1}\lambda,\text{if}\, a=1,P^{+}=h(P),\text{ if }\pi_{e}^{+}=\Phi(\pi_{e},1).
	\end{array}\right.
	\end{split}
	\end{equation}
	\item The one-stage cost function at time $ k$ is
	\begin{equation}
	\begin{split}
		&c_{k}(P_{k-1},\pi_{e,k-1},a_{k})
		\triangleq\beta\mathbb{E} [tr(P_{k})|P_{k-1},\pi_{e,k-1},a_{k}]\\&-(1-\beta)\sum\limits_{i=0}^{N}tr(h^{i}(P^{*}))\pi_{e,k}^{i}
		+ a_{k}\mathcal{C},
	\end{split}	
	\end{equation}
	where	$ \mathbb{E}[tr(P_{k})|P_{k-1},\pi_{e,k-1},a_{k}]	=(a_{k}\epsilon_{1}\lambda+(1-a_{k})\lambda)tr(P^{*})
	+(1-a_{k}\epsilon_{1}\lambda+(1-a_{k})\lambda)tr(h(P_{k-1})) $
	and  $ \pi_{e,k}
	 =\Phi(\pi_{e,k-1},a_{k}) $.
	
	\begin{remark}
		From Lemma 1, we obtain that  one-stage cost function  $ c_{k}(P_{k-1},\pi_{e,k-1},a_{k}) $ increases in $ P_{k-1} $.
	\end{remark} 
\end{enumerate}

Then, problem \eqref{e9} is equal to 
\begin{equation}\label{e26}
\min\limits_{a_{k}\in \{0,1\}}\sum\limits_{k=1}^{N}	c_{k}(P_{k-1},P_{e,k-1},a_{k}).
\end{equation} 

Similar to Section \ref{s3}, we will make the action $ a_{k+1} $ based on $ s_{k} $ instead of $ H_{k} $ by the Markov property proved in the same way as Theorem \ref{t1}, to minimize the total integrated cost. Define the optimality equation (Bellman equation) as
\begin{align*}
&V_{N+1}(P,\pi_{e})=0, \\
&V_{k}(P,\pi_{e})= \min\limits_{a\in \{0,1\}}\{c_{k}(P,\pi_{e},a)\\
&+\mathbb{P}(h(P),\Phi(\pi_{e},a)|s,a)V_{k+1}(h(P),\Phi(\pi_{e},a))\\
&+\mathbb{P}(P^{*},\Phi(\pi_{e},a)|s,a)V_{k+1}(P^{*},\Phi(\pi_{e},a))\},
\end{align*}
where $ s=(P,\pi_{e}) $ and  $ V_{k}(\cdotp,\cdotp) $ for $ k=1,2,\ldots,N $ is a real valued function.

%
\begin{theorem}
	$ V_{k}(P,\pi_{e}) $ is an increasing function in $ P $.  
\end{theorem}
\begin{proof}
	For fixed $ \pi_{e} $, we can use the same induction argument as in the proof of Theorem \ref{t2} to prove that the optimal solution $ V_{k}(P,\pi_{e}) $ is an increasing function in $ P $.
\end{proof}
\begin{theorem}\label{t6}
	For a fixed $ \pi_{e,k-1} $, the optimal solution to problem \eqref{e26} is a threshold policy on $ P_{k-1} $ of the form
	\begin{equation}
	a_{k}^{*}(P_{k-1},\pi_{e,k-1})=\left\lbrace
	\begin{array}{l}
	1, \text{if} \, P_{k-1}\leq h^{m(k)}(P^{*}), \\
	0, \text{otherwise} .
	\end{array}\right. 
	\end{equation}
	where the threshold $ m(k) \in\mathbb{N}  $ depends on $ k $ and $ \pi_{e,k-1} $
\end{theorem}
\begin{proof}
	Using the similar method in the proof of Theorem \ref{t3} in Appendix \ref{a2}, it is sufficient to prove that	
	\begin{equation}
		(1-\epsilon_{1}\lambda)V_{k+1}(h^{n}(P),\pi_{e})-(1-\lambda)V_{k+1}(h^{n}(P),\pi_{e}')
	\end{equation}	
	is a increasing function of $ P $ for all $ k=0,1,\ldots,N $,  $ n\in \mathbb{N} $ and  $ \pi_{e},\pi_{e}'\in \Pi_{e} $.	
	Therefore, for fixed $ k $ and $ \pi_{e} $, $ m(k)=\min\{t\in \mathbb{N} :\phi_{k}(h^{t}(P^{*}),\pi_{e})\geq 0\} $. 
\end{proof}

Theorem \ref{t6} guarantees the threshold structure without knowledge of the eavesdropper's estimation error covariance.

\section{simulation}\label{s4}
In this section, we use numerical examples to illustrate our optimal policies. Consider a system with $ A=\begin{bmatrix}
1.5&0  \\ 
0&0.9 
\end{bmatrix}  $, $ C=\begin{bmatrix}
1& 0
\end{bmatrix}  $, $  Q=\begin{bmatrix}
0.5&0  \\ 
0&0.5 
\end{bmatrix}  $, $R=0.6 $ , $ \lambda=\lambda_{e}=0.7 $, $ \epsilon_{1}=0.9$ and $ \epsilon_{2}=0.18 $. The normalized encryption cost is $ \mathcal{C}=6 $.  Set the weighted parameter $ \beta=0.5 $. 

Consider a finite time horizon with $ N=10 $. 
Fig. \ref{fig_known_k_5} plots $ a_{k}^{*} $ for different values of $ P_{k-1}=h^{n}(P^{*}) $ and $ P_{e,k-1}=h^{n_{e}}(P^{*})  $ at time step $ k=5 $. Fig.  \ref{fig_known_k_10} plots $ a_{k}^{*} $ at time $ k=10 $. Red blocks represent $ a_{k}^{*} =1$, while white ones represent $ a_{k}^{*} =0$. It is shown in these two figures that the threshold structure of the optimal policy $ a_{k}^{*} $ in both $ P_{k-1} $ and $ P_{e,k-1} $. Meanwhile, the optimal policy is dependent on time $ k $ as proved in Theorem \ref{t3}. 

\begin{figure}[!t]
	
	\subfigure[$ k=5 $.]{\label{fig_known_k_5}
		\begin{minipage}[!t]{0.48\linewidth}		
		\centering
		\includegraphics[width=1.6in]{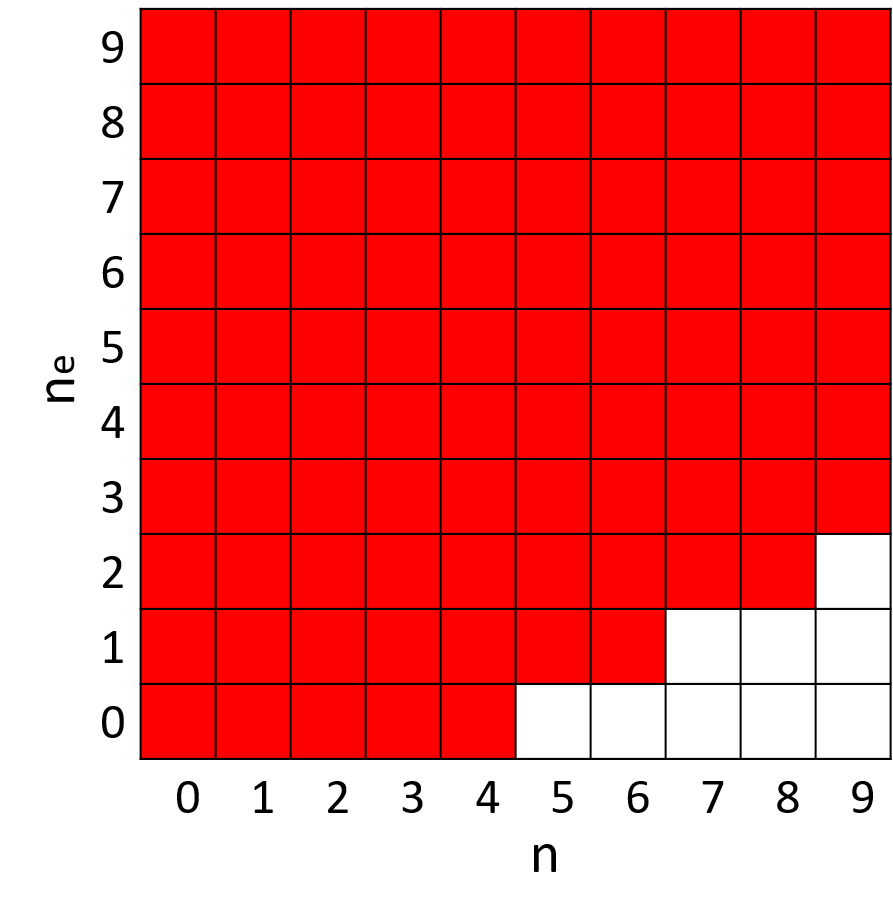}
		\end{minipage}%
	}
	\subfigure[$ k=10 $.]{\label{fig_known_k_10}
		\begin{minipage}[!t]{0.48\linewidth}		
			\centering
			\includegraphics[width=1.6in]{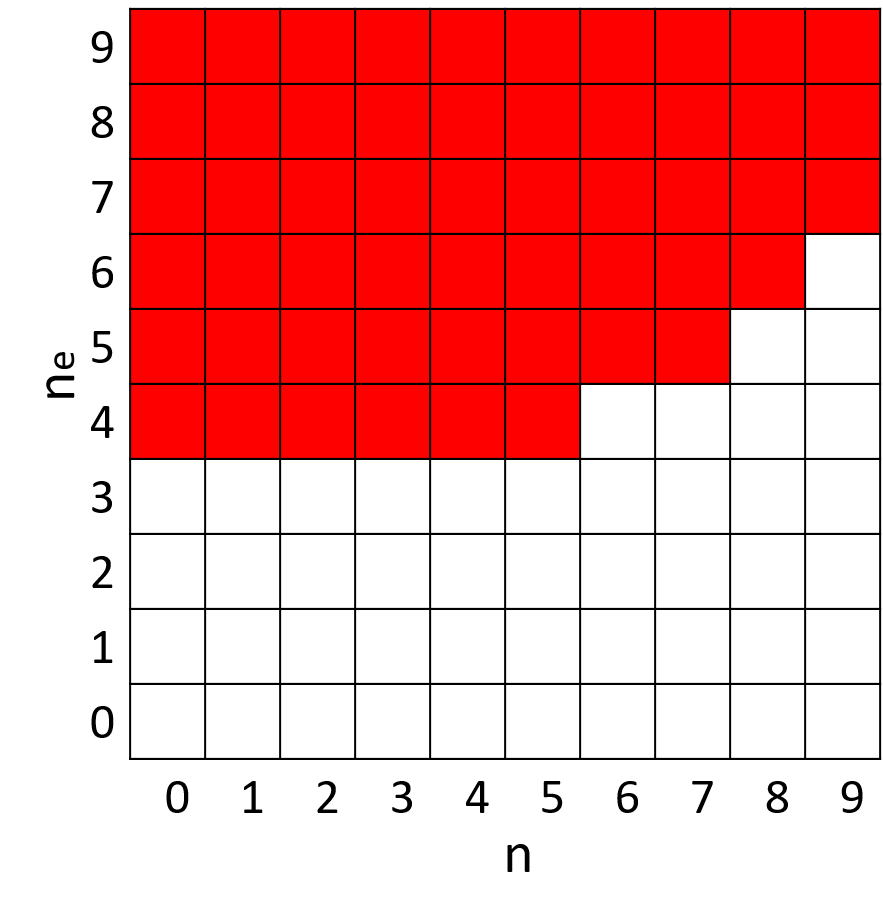}
		\end{minipage}%
	}
	\caption{Optimal policy at different time steps. 
	\label{f2}}
\end{figure}

 Furthermore, TABLE. \ref{table1} makes a comparison between the following four different encryption methods
 \begin{enumerate}
 	\item $ \theta^{1} $: always transmit the packet directly to the remote estimator without encryption, i.e., $ a_{k}=0 $, for all $ k $;
 	\item  $ \theta^{1} $: always encrypt the packet before each transmission, i.e., $ a_{k}=1 $, for all $ k $;
 	\item $ \theta^{*1} $: the optimal strategy derived from Section \ref{s3} with knowledge of the eavesdropper's estimation error covariance;
 	\item $ \theta^{*2} $: the optimal strategy derived from Section \ref{un} without knowledge of the eavesdropper's estimation error covariance.
 \end{enumerate} 
Consider the finite time horizon with $ N=6 $.  We run 1000 Monte Carlo tests.  We can see from TABLE. \ref{table1}  that the optimal encryption strategy reduces the total integrated cost significantly compared with using no encryption method and is better than encrypting all the messages. Furthermore, if we cannot obtain the exact error covariance of the eavesdropper, the optimal cost is larger than that when the error covariance is known to the remote estimator. 
 
 \begin{table}[h]
 	\caption{A comparison between encryption strategies}
 	\label{table1}
 	\begin{center}
 		\begin{tabular}{|c||c|c|c|}
 			\hline
 			 & $ \sum\limits_{k=1}^{N}\mathbb{E}[ tr(P_{k})] $&$ \sum\limits_{k=1}^{N}\mathbb{E}[ tr(P_{e,k})] $ & $ J_{k} $\\
 			\hline
 			$ \theta^{1} $ & 22.2487 & 22.2657 & -0.0085\\
 			\hline
 			$ \theta^{2} $& 24.1176 & 118.8861 & -11.3843\\ 			
 			\hline
 			$ \theta^{*1} $& 23.6582 & 114.0609 & -18.0513\\ 			
 			\hline
 			$ \theta^{*2} $& 23.8272& 118.2655& -12.3692\\ 			
 			\hline
 		\end{tabular}
 	\end{center}
 \end{table}
 
\section{Conclusion}\label{s5}
In this paper, we consider an optimal encryption schedule for a remote state estimation system in the presence of an eavesdropper. Our objective is to determine when to encrypt transmitted messages to minimize a linear combination of error covariance at the remote estimator and the eavesdropper, taking into account the cost of the encryption process.  This problem is shown to be formulated as a MDP, either with or without knowledge of the estimation error covariance at the eavesdropper. The optimal policy is proved to have a threshold structure in each situation. 

The current setup only focuses on the problem of a finite time horizon where the state space is finite.  It would be interesting to consider situations with a infinite time horizon.

\appendix
\subsection{Proof of Theorem \ref{t2}}\label{a1}
As $ V_{N+1}(P,P_{e})=0 $, it is trivial to see that $ V_{N+1}(P,P_{e}) $ is an increasing function in $ P $. Therefore, we prove the monotonicity using a backward induction way.

Assume that $ V_{t}(P,P_{e}) $ is increasing for $ t=k+1,\ldots,N+1 $, then we only need to prove $ V_{k}(P,P_{e}) $ is an increasing function in $ P $. We choose $ P'\geq P $,  one has $ h(P')\geq h(P)$. Denote $ s'\triangleq(P',P_{e}),s'_{00}\triangleq(h(P'),h(P_{e}) ),s'_{01}\triangleq(h(P'),P^{*} ) $. As function $ c_{k} $ increases in $ P $ and $\mathbb{P}(s'|s,a) $ is only dependent on action $ a $, we have
\begin{equation*}
\begin{split}
&V_{k}(P,P_{e})=\min\limits_{a\in \{0,1\}}\{c_{k}(P,P_{e},a)+\mathbb{P}_{00}(a)V_{k+1}(s_{00})+\\
&\mathbb{P}_{01}(a)V_{k+1}(s_{01})+\mathbb{P}_{10}(a)V_{k+1}(s_{10})+\mathbb{P}_{11}(a)V_{k+1}(s_{11})\}\\
&\leq c_{k}(P,P_{e},a_{s'}^{*})+\mathbb{P}_{00}(a_{s'}^{*})V_{k+1}(s_{00})+\mathbb{P}_{01}(a_{s'}^{*})V_{k+1}(s_{01})\\
&+\mathbb{P}_{10}(a_{s'}^{*})V_{k+1}(s_{10})+\mathbb{P}_{11}(a_{s'}^{*})V_{k+1}(s_{11})\\
&\leq c_{k}(P',P_{e},a_{s'}^{*})+\mathbb{P}_{00}(a_{s'}^{*})V_{k+1}(s_{00}')+\mathbb{P}_{01}(a_{s'}^{*})V_{k+1}(s_{01}')\\
&+\mathbb{P}_{10}(a_{s'}^{*})V_{k+1}(s_{10}')+\mathbb{P}_{11}(a_{s'}^{*})V_{k+1}(s_{11})=	V_{k}(P',P_{e})	.
\end{split}
\end{equation*}
The proof is completed.

We can use the same method to prove that $ V_{k}(P,P_{e}) $ is a decreasing function in $ P_{e} $. The proof is omitted.

\subsection{Proof of Theorem \ref{t3}}\label{a2}
Denote the difference of  $ V_{k}(P,P_{e}) $ when $ a^{*}=1 $ and $ a^{*}=0 $ as $ \phi_{k}(P,P_{e}) $. It can be calculated directly that
	\begin{align*}
	&\phi_{k}(P,P_{e})=  \beta(1-\epsilon_{1})\lambda(tr(h(P))-tr(h(P')))-\\
	&(1-\beta)(1-\epsilon_{2})\lambda_{e}(tr(h(P_{e}))-tr(h(P'))) +\mathcal{C}+\\
	&p_{1}V_{k+1}(s_{00})+p_{2}V_{k+1}(s_{01})+p_{3}V_{k+1}(s_{10})+p_{4}V_{k+1}(s_{11}),
	\end{align*}
	where $ p_{1}\triangleq \mathbb{P}_{00}(1)-\mathbb{P}_{00}(0)$, $p_{2}\triangleq  \mathbb{P}_{01}(1)-\mathbb{P}_{01}(0) $, $p_{3}\triangleq \mathbb{P}_{10}(1)-\mathbb{P}_{10}(0) $, $p_{4}\triangleq  \mathbb{P}_{11}(1)-\mathbb{P}_{11}(0) $. If $ \phi_{k}(P,P_{e})\geq 0 $, the optimal strategy at time $ k $ is $ a_{k}^{*}=0 $, otherwise $ a_{k}^{*}=1 $.
	
(1) It is equivalent to prove $  \phi_{k}(P,P_{e}) $ increases in $ P $ for fixed $ P_{e} $. Considering elements which relate to $ P $ in $ \phi_{k}(P,P_{e}) $, for $ P\geq P' $, we have 
\begin{align*}
&\phi_{k}(P,P_{e})-\phi_{k}(P',P_{e})=\beta(1-\epsilon_{1})\lambda(tr(h(P))-tr(h(P')))\\
&+p_{1}[V_{k+1}(h(P),h(P_{e}))-V_{k+1}( h(P'),h(P_{e}))]\\
&+p_{2}[V_{k+1}(h(P),P^{*})-V_{k+1}( h(P'),P^{*})].
\end{align*}

From Lemma 1, the first element $ \beta(1-\epsilon_{1})\lambda(tr(h(P))-tr(h(P')))\geq 0 $, it suffices to prove that
\begin{equation}
 p_{1}V_{k+1}(h(P),h(P_{e}))+p_{2}V_{k+1}(h(P),P^{*}),
\end{equation}	
is an increasing function of $ P $ for all $ k $. We will prove this statement using a backward induction way. We prove the slightly more general statement that	$ 
p_{1}V_{k+1}(h^{n}(P),P_{e})+p_{2}V_{k+1}(h^{n}(P),P_{e}') $
is an increasing function of $ P $ for all $ k=0,1,\ldots,N $, $ n\in \mathbb{N} $ and $ P,P_{e},P_{e}'\in \mathbf{S} $.

As $ V_{N+1}(P,P_{e})=0 $, it is trivial to see that the statement holds for $ k=N  $. Assume that $ \forall P\geq P'$,
\begin{equation}\label{e19}
\begin{split}
&p_{1}V_{t+1}(h^{n}(P),P_{e})+p_{2}V_{t+1}(h^{n}(P),P_{e}')\\
&-p_{1}V_{t+1}(h^{n}(P'),P_{e})-p_{2}V_{t+1}(h^{n}(P'),P_{e}')\geq 0
\end{split}		
\end{equation}
holds for $ t=k+1,\ldots,N $. Then	
\begin{align*}
&p_{1}V_{k}(h^{n}(P),P_{e})+p_{2}V_{k}(h^{n}(P),P_{e}')\\
&-p_{1}V_{k}(h^{n}(P'),P_{e})-p_{2}V_{k}(h^{n}(P'),P_{e}')\\&\geq\min\limits_{a\in \{0,1\}}	
\{p_{1}[c_{k}(h^{n}(P),P_{e},a)+\mathbb{P}_{00}(a)V_{k+1}(h^{n+1}(P),h(P_{e}))\\
&+\mathbb{P}_{01}(a)V_{k+1}(h^{n+1}(P),P^{*})-\mathbb{P}_{01}(a)V_{k+1}(h^{n+1}(P'),P^{*})\\
&-c_{k}(h^{n}(P'),P_{e},a)-\mathbb{P}_{00}(a)V_{k+1}(h^{n+1}(P'),h(P_{e}))]\\
&+p_{2}[c_{k}(h^{n}(P),P_{e}',a)+\mathbb{P}_{00}(a)V_{k+1}(h^{n+1}(P),h(P_{e}'))\\
&+\mathbb{P}_{01}(a)V_{k+1}(h^{n+1}(P),P^{*})-\mathbb{P}_{01}(a)V_{k+1}(h^{n+1}(P'),P^{*})\\
&-\mathbb{P}_{00}(a)V_{k+1}(h^{n+1}(P'),h(P_{e}'))-c_{k}(h^{n}(P'),P_{e}',a)]\}\\
&\geq\min\limits_{a\in \{0,1\}}	\{p_{1}[\mathbb{P}_{00}(a)V_{k+1}(h^{n+1}(P),h(P_{e}))\\
&+\mathbb{P}_{01}(a)V_{k+1}(h^{n+1}(P),P^{*})-\mathbb{P}_{01}(a)V_{k+1}(h^{n+1}(P'),P^{*})\\
&-\mathbb{P}_{00}(a)V_{k+1}(h^{n+1}(P'),h(P_{e}))]+p_{2}[\mathbb{P}_{00}(a)\\
&\cdot V_{k+1}(h^{n+1}(P),h(P_{e}'))+\mathbb{P}_{01}(a)V_{k+1}(h^{n+1}(P),P^{*})\\
&-\mathbb{P}_{00}(a)V_{k+1}(h^{n+1}(P'),h(P_{e}'))\\
&-\mathbb{P}_{01}(a)V_{k+1}(h^{n+1}(P'),P^{*})]\}\\
&=\min\limits_{a\in \{0,1\}}\{\mathbb{P}_{00}(a)[p_{1}V_{k+1}(h^{n+1}(P),h(P_{e}))+p_{2}\\
&\cdot V_{k+1}(h^{n+1}(P),h(P_{e}'))-p_{1}V_{k+1}(h^{n+1}(P'),h(P_{e}))-p_{2}\\
&\cdot V_{k+1}(h^{n+1}(P'),h(P_{e}'))]+\mathbb{P}_{01}(a)[p_{1}V_{k+1}(h^{n+1}(P),P^{*}))\\
&+p_{2}V_{k+1}(h^{n+1}(P),P^{*}))-p_{1}V_{k+1}(h^{n+1}(P'),P^{*}))\\
&-p_{2}V_{k+1}(h^{n+1}(P'),P^{*}))]\}\geq 0
\end{align*}
where the first inequality holds since $ a $ which denotes $ a_{k+1} $ is determined by the function $ \phi_{k+1} $. Meanwhile, the second inequality holds since $ c_{k}(h^{n}(P),P_{e},a) $ increases in $ P $ and the last inequality holds by the induction hypothesis \eqref{e19} and $ \mathbb{P}_{00}(a), \mathbb{P}_{01}(a)\geq 0$ for $ \forall a\in\mathbb{A} $.

Therefore, for fixed $ k $ and $ P_{e} $, $ m(k)=\min\{t\in \mathbb{N} :\phi_{k}(h^{t}(P^{*}),P_{e})\geq 0\} $.

(2)	As $ -tr(h(P_{e}) )$ decreases in $ P_{e} $, it is equivalent to prove that	
$ p_{1}V_{k+1}(h(P),h(P_{e}))+p_{3}V_{k+1}(P^{*},h(P_{e}))  $
is a decreasing function of $  P_{e} $. Similar to the first part,  we prove by induction the slightly more general statement	
$ p_{1}V_{k+1}(P,h^{n}(P_{e}))+p_{3}V_{k+1}(P',h^{n}(P_{e})) $
is a decreasing function of $ P_{e} $ for all $ k=0,1,\ldots,N $, $ n\in \mathbb{N} $ and $ P,P',P_{e}\in \mathbf{S} $. The details are omitted.	





%

\bibliographystyle{IEEEtran}
\bibliography{ref}

\end{document}